\documentclass[a4paper,12pt]{article}

\usepackage{charter, fullpage}
\usepackage[english]{babel}
\usepackage{vmargin}
\usepackage{calc}
\usepackage[latin1]{inputenc}
\usepackage{amsfonts,amsmath,amstext,amsthm,amssymb}
\usepackage{pifont}

\setpapersize{A4}
\setmarginsrb{30mm}{20mm}{30mm}{20mm}{12pt}{11mm}{0pt}{11mm}

\newtheorem{fact}{Fact}[section]
\newtheorem{definition}{Definition}[section]
\newtheorem{lemma}{Lemma}[section]
\newtheorem{corollary}{Corollary}[section]
\newtheorem{theorem}{Theorem}[section]

\newcommand*{\eps}{\epsilon}

\newcommand{\vcent}[1]{\raisebox{1.5ex}[-1.5ex]{#1}}
\newcommand{\vcentmath}[1]{\raisebox{1.0ex}[-1.0ex]{#1}}
\newcommand*{\bs}{\backslash}

\def\polylog{\operatorname{polylog}}

\title{\textbf{Approximating Subdense Instances\\ of Covering Problems
}}

\author{
  Jean Cardinal\thanks{Université Libre de Bruxelles (ULB), CP212. \texttt{Email:~jcardin@ulb.ac.be}}\\[.75ex]
  Richard Schmied\thanks{Dept. of Computer Science, University of Bonn.
    Work supported by Hausdorff Doctoral Fellowship.
    Research partially done while visiting Princeton University.
    \texttt{Email:~schmied@cs.uni- bonn.de}}
  \and
  Marek Karpinski\thanks{Dept. of Computer Science and the Hausdorff
    Center for Mathematics, University of Bonn.
    Supported in part by DFG grants and the Hausdorff Center grant EXC59-1.
    \texttt{Email:~marek@cs.uni- bonn.de}}\\[.75ex]
  Claus Viehmann\thanks{Dept. of Computer Science, University of Bonn.
    Work partially supported by Hausdorff Center for Mathematics, Bonn.
    Research partially done while visiting Princeton University.
    \texttt{Email: viehmann@cs.uni-bonn.de}}
}
\date{}
\begin{document}
\maketitle
\begin{abstract}
We study approximability of \emph{subdense} instances of various
covering problems on graphs, defined as instances in which the minimum or average degree
is $\Omega (n/\psi (n))$ for some function $\psi(n) =\omega (1)$ of the instance size. 
We design new approximation algorithms as well as new polynomial time approximation schemes (PTASs)
for those problems and establish first approximation hardness results for them.
Interestingly, in some cases we were able to prove optimality of the
underlying approximation ratios, under usual complexity-theoretic assumptions.
Our results for the Vertex Cover problem depend on an improved recursive sampling method
which could be of independent interest.\\

\end{abstract}


\section{Introduction}

In order to cope with the intractability of combinatorial optimization problems, it is often useful
to consider specific families of instances, the properties of which can be exploited to obtain polynomial-time
algorithms with improved approximation guarantees. {\em Dense} instances of optimization problems on graphs, 
in which the minimum or the average degree of the input graph is high, are an example of such a family. 
Typically, an instance of a combinatorial optimization on a graph is said to be dense whenever 
the minimum or average degree of the graph is $\Omega (n)$, where $n$ is the number of vertices.
It was proven in Arora, Karger and Karpinski~\cite{AKK95} that 
a wide range of maximization problems restricted to dense instances admit 
polynomial time approximation schemes (PTAS). Later, Karpinski and Zelikovsky~\cite{KZ97} 
defined and studied dense cases of covering problems. In particular,
they investigated the approximation complexity of dense versions of Set Cover,
Vertex Cover and the Steiner Tree problem. Density parameters 
have been used in approximation ratios for various optimization problems
(see \cite{K01} for a detailed survey and \cite{KZ97,II05,CL10} for the Vertex Cover problem).\\

In this paper, we aim at bridging the gap between dense and sparse instances, by considering
wider ranges for the density parameters. We define families of {\em subdense} instances, in which
the minimum (or average) degree of the underlying graph is $o(n)$, but $\Omega (n/\psi (n))$ for some
function $\psi$ of the instance size $n$. We consider subdense cases of covering and Steiner Tree problems, 
for which previous results have been obtained in dense cases.
We now describe precisely the problems and our contributions.

\subsection{Problems and Previous Results}

\paragraph{Vertex Cover.}

The \emph{Vertex Cover} (VC) \emph{problem} asks to find for a given graph
$G=(V,E)$ a minimum size vertex set $C\subseteq V$ which covers the edges of $G$.
We call a graph $G=(V,E)$ ($n=|V|$) \emph{everywhere $\eps$-dense} for $\eps>0$ if all
vertices in $G$ have at least $\eps n$ neighbors, and 
{\em average $\eps$-dense} if the average degree of a vertex in $G$ is at least
$\eps n$ (that is, the number of edges is a constant fraction of that of the complete graph).
Currently, the best parameterized ratios for the Vertex Cover problem for average and
everywhere dense instances are $2/(2-\sqrt{1-\eps})$ and $2/(1+\eps)$, respectively \cite{KZ97}.
Imamura and Iwama~\cite{II05} later improved the former result, by generalizing it to
depend on both average degree $\bar{d}$ and the maximum degree $\Delta$.
In particular, they gave an approximation algorithm with an approximation
factor smaller than 2, provided that the ratio $\bar{d}/\Delta > c$ for some constant $c<1$, 
and $\Delta=\Omega (n \frac{\log \log n}{\log n})$. 

As for lower bounds, Clementi and Trevisan~\cite{CT99} as well as Karpinski and Zelikovsky~\cite{KZ97} 
proved that the Vertex Cover problem restricted to everywhere and average dense graphs remains APX-hard.
Later, Eremeev~\cite{E99} showed that it is NP-hard to approximate the Vertex Cover problem
in everywhere $\eps$-dense graphs within a factor less than $(7+\eps)/(6+2\eps)$.
Finally, Bar-Yehuda and Kehat \cite{BK04} prove that if the Vertex Cover problem cannot be
approximated within a factor strictly smaller than 2 on arbitrary graphs, then it cannot
be approximated within factors smaller than $2/(2-\sqrt{1-\eps})-o(1)$ and
$2/(1+\eps)-o(1)$, respectively, on average and everywhere dense graphs.
Imamura and Iwama~\cite{II05} later proved that their obtained approximation bound is best possible 
unless the general vertex cover problem can be approximated within $2-\epsilon$ for some $\epsilon > 0$.\\

\paragraph{Connected Vertex Cover.}

The \emph{Connected Vertex Cover} (CVC) \emph{problem} is a variant of the VC problem,
the goal of which is to cover the edges of a given connected graph with a minimum size vertex
set that induce a connected subgraph. Cardinal and Levy~\cite{CL10} gave 
the first parameterized approximation ratios for the CVC problem. The achieved ratios are $2/(2-\sqrt{1-\eps})$
for the average dense case and $2/(1+\eps)$ for the everywhere dense case.
Assuming the Unique Game Conjecture~\cite{KR08}, the reduction in \cite{BK04} implies that 
the achieved ratios are the best possible. Cardinal and Levy conjectured
that similar results due to Imamura and Iwama can be obtained in subdense cases
and left it for future work.\\

\paragraph{Set Cover.}

A Set Cover instance consists of a ground set $X$ of size $n$ and a collection $P\subset 2^X$ 
of $m$ subsets of $X$. The goal is to find a minimum size subset of $P$ that covers $X$, that is, a minimum 
cardinality subset $M\subseteq P$ such that $\bigcup_{x\in M}x = X$. 
The simple greedy heuristic for the Set Cover problem was already 
studied by Johnson~\cite{J74} and Lov{\'a}sz~\cite{L75}. They showed independently 
that this approximation algorithm provides a $(1+\ln(n))$-approximate
solution which is essentially the best achievable under reasonable complexity-theoretic
assumptions \cite{F98}. 

Bar-Yehuda and Kehat~\cite{BK04} studied dense versions of the Set Cover problem. Their definition of a dense instance
generalizes that of the dense Vertex Cover problem; they study $k$-bounded instances, in which every element of $X$
appears in at most $k$ subsets of $P$ (for the Vertex Cover problem, we have $k=2$), but where $n \geq \eps m^k$. 
They generalized the results of Karpinski and Zelikovsky~\cite{KZ97} to this setting.

On the other hand, Karpinski and Zelikovsky~\cite{KZ97} defined $\epsilon$-dense instances of the Set Cover problem
as instances in which every element of $X$ belongs to a fraction at least $\epsilon$ of the subsets of $P$.
They showed that such dense instances can be approximated within $c\ln(n)$ for every $c>0$
in polynomial time and can be solved to optimality in time $m^{O(\ln(n))}$.\\

\paragraph{Steiner Tree.}

In the \emph{Steiner Tree problem} (STP), we are given a connected graph $G=(V,E)$
and a set of distinguished vertices $S\subseteq V$, called \emph{terminals}.
We want to find a minimum size tree within $G$ that spans all terminals from $S$.
The Steiner Tree problem is one of the first problems shown to be NP-hard by Karp~\cite{K75}.
A long sequence of papers give approximation factors better than 2
\cite{Z93,BR92,KZ97a,HP99,RZ05,BGRS10}.
The current best approximation ratio is 1.39 due to Byrka et al.~\cite{BGRS10}.
Since the problem is APX-hard \cite{BP89}, we do not expect to
obtain a PTAS for this problem in arbitrary graphs.
In particular, it is NP-hard to find solutions of cost less than $\frac{96}{95}$ time
of the optimal cost~\cite{CC08}.

Karpinski and Zelikovsky~\cite{KZ97} introduced the $\eps$-dense Steiner Tree problem. 
A Steiner Tree is called $\eps$-dense if each terminal $t\in S$ has at least
$\eps |V\bs S|$ neighbors in $V\bs S$. They showed that for every $\eps>0$ the $\eps$-dense Steiner Tree problem admits
a PTAS. Later on, Hauptmann~\cite{H07} showed that the same scheme even yields an efficient PTAS for the
$\eps$-dense Steiner Tree problem.\\

\subsection{Our Contributions}

In this work, we consider a natural extension of the $\eps$-density: an instance is called $\psi(n)$-dense, if 
the minimum (or average) degree is at least $\frac{n}{\psi (n)}$.
In the remainder, we require $\psi(n)$ to be $\omega(1)$ and $n$ is a size parameter 
of the instance. This type of instances are called \emph{nondense} and were considered before for the MAX-CUT 
problem by 
Fernandez de la Vega and Karpinski~\cite{FK06}.\\

In Section~\ref{sec:vc}, we show how to modify the sampling procedure of the Imamura-Iwama algorithm (II-algorithm) \cite{II05})  
so that the degree condition can be slightly relaxed. This yields the same approximation factors in subdense instances where
$\psi (n)$ is slightly sublogarithmic. More precisely, given a graph $G$ with average degree $\bar{d}$, maximum degree $\Delta\leq \frac{n}{2}$ and 
$\Delta =\Omega (n / \psi(n))$, we can compute with high probability a solution with approximation ratio 
$\frac{2}{1+\bar{d}/(2\Delta)}$ in $n^{O(1)} 2^{O(\psi(n)\log\log\log(n))}$ time which is polynomial
if $\psi (n) = O(\log(n)/\log \log \log(n))$.  
From the previously known hardness result, the approximation ratio is best possible unless Vertex Cover can be 
approximated within $2-\eps$ for some $\eps > 0$.\\

In Section~\ref{sec:cvc}, we introduce the Subset Connected Vertex Cover (SCVC) problem which
generalizes the CVC problem and present a constant-factor approximation algorithm for the problem. 
This algorithm relies on a previously known approximation algorithm for quasi-bipartite instances of the Steiner Tree problem.
By combining the approximation algorithm for the SCVC problem with the modified II-algorithm, 
we prove the existence of a randomized approximation algorithm with the same approximation factor and running time
as for Vertex Cover in nondense instances. In particular, given a graph $G$ with average degree $\bar{d}$, maximum degree $\Delta\leq \frac{n}{2}$ and 
$\Delta =\Omega (n / \psi(n))$, we can compute with high probability a connected vertex cover with approximation ratio 
$\frac{2}{1+\bar{d}/(2\Delta)}$ in $n^{O(1)} 2^{O(\psi(n)\log\log\log(n))}$ time 
which is polynomial
if $\psi (n) = O(\log(n)/\log \log \log(n))$ and quasipolynomial
if $\psi (n) = \polylog(n)$. \\ 

In Section~\ref{sec:sc},
we study $\psi (n)$-dense instances of the Set Cover problem, in which
every element of the ground set $X$ is contained in at least $|P|/\psi (n)$ subsets of the family $P$. 
In particular, we analyze the performance of the greedy heuristic and show that
the greedy algorithm for $\psi (n)$-dense instances of the Set Cover problem returns a solution of size at most $\psi (n) \ln n$. Consequently, it follows that $\psi (n)$-dense instances 
can be solved exactly in time $O(|P|^{\ln(n)\psi(n)})$ 
and unless $NP\subseteq DTIME [n^{\log n \cdot \psi(n)}]$, the $\psi(n)$-dense
Set Cover problem is not NP-complete.
This algorithm is used as a subroutine for the approximation algorithm for the
Steiner Tree problem.\\

The subdense Steiner Tree problem is tackled in Section~\ref{sec:st}. 
An instance of the Steiner Tree problem is said to be $\psi(n)$-dense when
every terminal in $S$ is connected to at least $|V\setminus S|/\psi (n)$ nonterminals. 
We prove that these instances can be approximated within $1+\delta$ in time $n^{O(1)}2^{O(\frac{\psi(n)}{\delta})}$. 
In particular, it yields a PTAS for the cases where $\psi(n)=O(\log n)$, and a QPTAS for the cases where $\psi (n) = \polylog n$. 
On the negative side, we show that for every $\delta,\eps>0$, the Steiner Tree problem restricted to 
$|V\setminus S|^{1-\delta}$-everywhere dense graphs is NP-hard to approximate with $\frac{263}{262}-\eps$. 
Our results are listed in the Tables~\ref{tab:over:upp} and~\ref{tab:over:low}.  
(We assume $\psi(n)\geq 2$.)


\begin{samepage}
\begin{table}[h]
\centering
\begin{tabular}{@{}||c|c|c|c||@{}}
  \hline \hline
  \multicolumn{4}{||c||} {\textbf{Upper Bounds}} \\
  \hline
  \hline
  \textbf{VC} & \textbf{CVC} & \textbf{SC} & \textbf{STP} \\
  \hline
  \hline
    \multicolumn{4}{||c||} {Subdense: $\psi(n)=O(\log n )$} \\
 \hline
 \hline
 & & & \\
  \vcent{$\psi(n)\!=\!O(\frac{\ln n}{\ln\ln\ln n})$} &
  \vcent{$\psi(n)= O (\frac{\ln n}{\ln\ln\ln n})$} & 
  \vcent{$\psi(n) = O(\log n)$} &
  \vcent{$\psi(n) = O(\log n)$} \\ 
%
\hline
 & & & \\
  \vcent{$\frac{2}{1+\frac{\bar{d}}{2\Delta}}$-approx.} &
  \vcent{$\frac{2}{1+\frac{\bar{d}}{2\Delta}}$-approx.} &
  \vcent{exact}  &
  \vcent{PTAS}
  \\
\hline
  poly & 
  poly &
  $O\left(m^{O(\log^2 n)}\right)$ &
  \\
  time &
  time &
  time &
  \vcent{$-$} \\
  \hline
  \hline
  \multicolumn{4}{||c||} {Mildly Sparse: $\psi(n)=\polylog n$} \\
  \hline
  \hline
  & & & \\
  \vcent{same as in \cite{II05}}  & 
  \vcent{$\frac{2}{1+\frac{\bar{d}}{2\Delta}}$-approx.} &
  \vcent{exact} &
  \vcent{QPTAS} \\
\hline
   &
  quasipoly  &
  $O(m^{\polylog n})$  &
  \\
  \vcent{$-$} &
  time &
  time &
  \vcent{$-$}  \\
  \hline
  \hline
  \multicolumn{4}{||c||} {Nondense: $\psi(n)=\omega(1)$} \\
\hline
\hline
  & & & \\
  \vcent{$\frac{2}{1+\frac{\bar{d}}{2\Delta}}$-approx.} & 
  \vcent{$\frac{2}{1+\frac{\bar{d}}{2\Delta}}$-approx.} &
  \vcent{exact} &
  \vcent{$(1+\delta)$-approx.} \\
\hline
  $n^{O(1)} 2^{O(\psi(n)\ln\ln\ln n)}$ &
  $n^{O(1)} 2^{O(\psi(n)\ln\ln\ln n)}$ &
  $O(m^{\psi(n)\ln n})$  &
  $n^{O(1)}2^{O(\frac{\psi(n)}{\delta})}$
  \\
  time & time & time & time \\
  \hline \hline
\end{tabular}
\caption{Overview Upper Bounds\label{tab:over:upp}}
\end{table}

\begin{table}[h]
\centering
\begin{tabular}{||p{.212 \linewidth}|p{.212 \linewidth}|p{.212 \linewidth}|c||} 
  \hline \hline 
 \multicolumn{4}{||c||} {\textbf{Lower Bounds}} \\
  \hline
  \hline
  \centering\textbf{VC} & \centering\textbf{CVC} & \centering\textbf{SC} & \textbf{STP} \\
  \hline
  \hline
   \centering $\forall \delta>0$&
   \centering $\forall \delta>0$&
   &
    $\forall \delta>0$
   \\
  \centering $\psi(n)=O(n^{1-\delta})$& 
  \centering $\psi(n)=O(n^{1-\delta})$&
  \centering \vcent{$-$}&
  $\psi(n)=|V\bs S|^{1-\delta}$
  \\
  \hline
  \centering$\forall \epsilon>0$
  & \centering $\forall \epsilon>0$&
  &$\forall \epsilon>0$
  \\
    \centering UGC-hard with
  & \centering UGC-hard with 
  & 
  &  NP-hard with
  \\
  \centering $\frac{2}{1+\frac{\bar{d}}{2\Delta}}-\epsilon$&
  \centering $\frac{2}{1+\frac{\bar{d}}{2\Delta}}-\epsilon$&
  \centering \vcent{$-$}&
  $\frac{263}{262}-\eps$
  \\
  \centering (cf. \cite{II05})& \centering (cf. \cite{II05})&& $ $\hspace{.2 \linewidth} $ $\\
  \hline \hline
\end{tabular}
\caption{Overview Lower Bounds\label{tab:over:low}}
\end{table}
\end{samepage}


\section{The Subdense Vertex Cover Problem}
\label{sec:vc}

Previous hardness results in Imamura and Iwama~\cite{II05} and Karpinski and Zelikovsky~\cite{KZ97} imply that approximating subdense instances -- in which either the average or the minimum degree is $o(n)$ -- within a ratio strictly smaller than 2 is as hard as approximating the general minimum Vertex Cover problem within a ratio smaller than 2. 

However, if we make the hypothesis that the average and maximum degree, denoted respectively by $\bar{d}$ and $\Delta$, are of the same order, we can obtain better approximation results in the subdense cases. Imamura and Iwama~\cite{II05} gave an approximation algorithm with an approximation factor smaller than 2 for vertex cover, provided that the ratio $\bar{d}/\Delta > c$ for some constant $c<1$, and both are $\Omega (n \frac{\log \log n}{\log n})$. They show that the obtained approximation bound is best possible unless the general Vertex Cover problem can be approximated within $2-\epsilon$ for some $\epsilon > 0$. In what follows, we show how to modify the sampling procedure of the II-algorithm so that the degree condition can be slightly relaxed. Before that, we first outline their algorithm and state their main results.

\paragraph{The Imamura-Iwama Algorithm.}

The algorithm is a refinement of the original Karpinski-Zelikovsky method, in which a large subset $W$ of a minimum vertex cover is computed, then a 2-approximation is computed on the graph induced by the remaining uncovered edges. In general, it is sufficient to compute a subset $W$ that has a large intersection with a minimum vertex cover. The following lemma bounds the approximation ratio obtained in this way as a function of the number of vertices of $W$ that are contained in an optimal solution $C$.

\begin{lemma}[\cite{II05}]
\label{lem:subset}
Let $W$ be a subset of $V$. If there exists a minimum vertex cover $C$ such that $|W\cap C| = n_1$ and $|W\setminus C|=n_2$, then we can construct in polynomial time a vertex cover $C'$ such that $|C'| \leq \frac{2}{1+\frac{n_1 - n_2}{n}} |C|$. 
\end{lemma}

The II-algorithm uses both randomization and recursion to construct a suitable subset $W$. 
It is given in Figure~\ref{fig:I-I}, where:
$$
r(G) := |V|\left( 1-\sqrt{1-\frac{\bar{d}}{|V|}}\right) .
$$
The initial call to the algorithm is made with $i=1$ and a well-chosen value for $t$, as described below.

\begin{figure}
  \fbox{\hspace*{.4cm}
    \begin{minipage}{\textwidth - .9cm}
      {\bf II-Algorithm ${\cal A}(t, i, G)$}\\[1.5ex]
      {\bf Input:} Graph $G=(V,E)$, integers $t$ and $i$, \\[1ex]
      \hspace*{.5cm} $s := 2(\log n)^2$\\
      \hspace*{.5cm} {\bf if} $i<t$ {\bf then}\\
      \hspace*{1cm} let $H := \{v\in V : d(v)\geq r(G) \}$\\
      \hspace*{1cm} let $U_G$ be a set of $s$ vertices from $H$ chosen uniformly at random\\
      \hspace*{1cm} let ${\cal V}_G := \{ H \}\cup \{ N(v) : v\in U_G\}$\\
      \hspace*{1cm} {\bf for each} $V_j\in {\cal V}_G,\ 1\leq j\leq s+1$ {\bf do}\\
      \hspace*{1.5cm} $C_j := V_j\cup {\cal A} (t, i+1, G - V_j)$\\
      \hspace*{1cm} {\bf end for}\\
      \hspace*{1cm} {\bf return} a minimum size set among $\{C_j\}_{j=1}^{s+1}$\\   
      \hspace*{.5cm} {\bf else} ($i=t$)\\
      \hspace*{1cm} {\bf return} a 2-approximation of the minimum vertex cover of $G$\\
      \hspace*{.5cm} {\bf end if}\\
      \medskip
    \end{minipage}
  }
  \caption{\label{fig:I-I}Imamura-Iwama algorithm.}
\end{figure}

\paragraph{Main Results.}

Let $P$ be a path in the recursion tree of the II-algorithm. We denote by $W_P$ the set of removed vertices 
corresponding to this path, that is, $W_P := \cup_{i=1}^{|P|} W_i^P$, where $W_i^P$ is the set $V_j$ chosen at
the $i$-th step on the path $P$.
\begin{lemma}
Let $C$ be a minimum vertex cover of $G$. Then with high probability there is a path $P$ 
in the recursion tree such that $|W_P \cap C|\geq (1-o(1))|W_P|$.
\end{lemma}
A path satisfying the above condition is called a {\em good} path.

To compute the approximation ratio obtained by the II-algorithm, we need to define the function $\gamma (G)$. 
The function $\gamma (G)$ is equal to $\frac 1n$ times the minimum size of an optimal vertex cover in a graph with the same number of vertices, average, and maximum degree as $G$. 
This function can be computed explicitly:
$$
\gamma(G) = \left\{ \begin{array}{c@{\quad if \quad}l}
\frac{\bar{d}}{2\Delta} & |E| \leq \Delta(n-\Delta) \\
\frac{n+\Delta -\sqrt{(n+\Delta)^2-4\bar{d}n}}{2n} & |E| > \Delta(n-\Delta)
\end{array} \right.
$$ 
The following lemma indicates that the size of $W$ can always match this lower bound, provided the depth of the recursion is sufficient.
\begin{lemma}
\label{lem:depth}
For any path $P$ of height $t$, $|W_P|\geq (1-(1-\frac{\Delta}{4n})^t) \gamma (G) n$.
\end{lemma}
We wish to compute a subset $W$ of size arbitrarily close to $\gamma (G) n$. From the previous lemma, this requires the number of iterations 
to be $t = \Theta (n/\Delta )$. The size of the recursion tree is $(s+1)^t$, which is polynomial only if $\Delta = \Omega (n \log \log n / \log n)$.

From Lemma~\ref{lem:subset}, the corresponding approximation factor is then arbitrarily close to $\frac{2}{1+\gamma (G)}$.
In the range of
$\Delta\leq \frac{n}{2}$, the inequality $|E|\leq \Delta (n-\Delta )$ is always satisfied, since $|E|=\frac{\bar{d}n}{2}\leq\frac{\Delta n}{2}\leq\Delta(n-\Delta)$. 
Thus, we have $\gamma (G) = \bar{d} / (2\Delta)$ and the result follows.

\begin{theorem}
There is a randomized polynomial-time approximation algorithm with an approximation factor arbitrarily close to $\frac 2{1+\left. \frac{\bar{d}}{2\Delta}\right. }$ for the Vertex Cover problem in graphs with average degree $\bar{d}$
and maximum degree $\Delta$, such that $\Delta\leq\frac{n}{2}$ and $\Delta = \Omega(n \log \log n / \log n)$.
\end{theorem}
\paragraph{A Modified Sampling Procedure.}

We now describe a modified version of the II-algorithm, in which the sample size $s$ is reduced.
We denote by $C$ a minimum vertex cover, by $W_i$ the set obtained at the $i$-th step in
a good path, and by $W$ the set $W := \cup_i W_i$.

\begin{lemma}
$|W_i \cap C| > (1 - 1/ \sqrt{s})|W_i|$ with probability at least $1-1/e^{\sqrt{s}}$.
\end{lemma}
\begin{proof}
If a fraction at least $1-1/\sqrt{s}$ of $H$ intersects $C$, then we are done. Otherwise, a random vertex $v$ of $H$ is such that $v\in C$ with probability at most $1-1/\sqrt{s}$. Hence with probability $1-(1-1/\sqrt{s})^s \to 1-1/e^{\sqrt{s}}$ (for sufficiently large values of $n$) we get a vertex $v\not\in C$ in our sample. The neighborhood $N(v)$ of this vertex must therefore be completely contained in $C$.   
\end{proof}
Multiplying the probabilities along the path yields the following:
\begin{lemma}
\label{lem:prob}
With probability $(1-1/e^{\sqrt{s}})^t$, we have $|W\cap C|>(1 - 1/\sqrt{s})|W|$.
\end{lemma}

We now analyze the behavior of a modified version of the II-algorithm, in which we set the two parameters $s$ and $t$ to the following values:
\begin{eqnarray}
s & := & (\ln \ln n - \ln \ln \ln \ln n)^2 \\
t & := & a e^{\sqrt{s}} .
\end{eqnarray}

First, we check that the probability of having an arbitrarily high proportion of $W$ belonging to $C$ is constant. From Lemma~\ref{lem:prob}, this probability is
$$
(1-1/e^{\sqrt{s}})^t = (1 - 1/e^{\sqrt{s}})^{a e^{\sqrt{s}}} = \left((1 - 1/x)^x\right)^a \to 1/e^a ,
$$
where $x=e^{\sqrt{s}}$ and we assume that $n$ is sufficiently large. Thus, for these parameters, the algorithm computes a fraction $1-o(1)$ of $W$ is contained in
an optimal vertex cover with a probability close to $1/e^a$.\\

Next, we show that the algorithm runs in polynomial time. The size of the search tree is $(s+1)^t = (s+1)^{a e^{\sqrt{s}}}$, hence we require: 
\begin{eqnarray}
(s+1)^{a e^{\sqrt{s}}} & = & \mathrm{poly}(n) \\
\sqrt{s} + \ln \ln s & < & \ln \ln n + O(1) .
\end{eqnarray}
For the chosen value of $s$, we obtain $e^{\sqrt{s}} = \frac {\ln n}{\ln \ln \ln n}$, 
and $\sqrt{s} + \ln \ln s \leq \ln \ln n - \ln \ln \ln \ln n + \ln \ln \ln \ln n + O(1) = \ln \ln n + O(1)$. 
So for this sample size, the running time of the algorithm remains polynomial.\\

Finally, from Lemma~\ref{lem:depth}, the desired approximation factor is obtained only when $t = \Theta\left( \frac n\Delta \right)$. This yields: 
$$
\Delta = \Omega \left( \frac nt\right) = \Omega \left( \frac n{e^{\sqrt{s}}} \right) = \Omega \left( n \cdot \frac {\ln \ln \ln n} {\ln n} \right) .
$$
In general, the parameterized running time with $\psi(n)$ as parameter can be
expressed by the following:
$$n^{O(1)}(s+1)^t=n^{O(1)}(s+1)^{O(\frac{n}{\Delta})}\leq n^{O(1)}(\ln\ln(n))^{O(\frac{n}{\Delta})} = n^{O(1)}2^{O(\psi(n)\ln\ln\ln(n))}$$
Taking into account the above constraint, we also get a quasi-polynomial time algorithm in instances with $\Delta = \Omega (n/\polylog n)$.

\begin{theorem}\label{thm:mainvc}
There is a randomized approximation algorithm with an approximation factor arbitrarily close to $2/(1+\bar{d}/(2\Delta))$ for the Vertex Cover problem in graphs with average degree $\bar{d}$, 
maximum degree $\Delta\leq n/2$ and $\Delta=\Omega(n/\psi(n))$ 
with
running time $n^{O(1)} 2^{O(\psi(n)\log\log\log(n))}$ which 
is
polynomial if $\psi (n) = O\left(\frac{\log n}{\log \log \log n}\right)$.
\end{theorem}


\section{The Subdense Connected Vertex Cover Problem}
\label{sec:cvc}

In this section, we study the approximability of nondense instances of the Connected 
Vertex Cover (CVC) problem. Furthermore, we obtain an approximation algorithm for a
generalization of the CVC problem.

Cardinal and Levy \cite{CL10} studied connected vertex covers in everywhere and
average dense graphs. For both cases, they gave approximation algorithms
with asymptotically optimal approximation ratios assuming the Unique Game Conjecture.
In the same work, they conjectured that the Connected Vertex Cover problem
restricted to even less dense graphs
can be approximated with an approximation ratio smaller than $2$ in 
polynomial time and left it for future work.

We introduce the Subset Connected Vertex Cover (SCVC) problem which
generalizes the CVC problem and present a constant-factor
approximation algorithm. By combining the approximation
algorithm for the SCVC problem with the modified II-algorithm, we prove the existence of a randomized 
approximation algorithm with approximation ratio less than $2$ for the CVC problem 
in a restricted class of subdense graphs running in polynomial time. If we allow quasi-polynomial
running time, we obtain the same approximation ratio for a larger class of subdense instance.    

\subsection{The Subset Connected Vertex Cover Problem}   

\begin{definition}{Subset Connected Vertex Cover problem (SCVC)}\\
Given a connected graph $G=(V,E)$ and a set $S\subset V$, a feasible solution consists of a 
connected vertex cover $Y$ of $G$ such that $S\subseteq Y$.
The task is to find such a $Y$ with minimum cardinality.
\end{definition}

The maximal matching heuristic is a simple $2$-approximation algorithm for
the Vertex Cover problem. It consists of choosing all vertices in any maximal matching.
In order to obtain a $2$-approximate solution for the CVC problem, we have to
extend this concept:
For a given connected graph $G=(V,E)$, a \emph{connected maximal matching}
(CMM) is a maximal matching $M$ in $G$ such that $G[V(M)]$ is connected.

Clearly, $V(M)$ is a $2$-approximate solution for the CVC problem and 
$M$ can be computed in polynomial time.

As a subroutine in our approximation algorithm for the SCVC problem,
we will use an approximation algorithm for
the Steiner Tree problem for special instances called quasi-bipartite:
An instance of the Steiner Tree problem $G=(V,E)$ and  terminal set $S$
is called \emph{quasi-bipartite}, if $V\setminus S$ forms an independent set in $G$.
Gr{\"o}pl et al.~\cite{GHNP02} gave an $1.217$-approximation algorithm for
quasi-bipartite instances, which we will note as $\mathcal{A}_{qbST}$ in the following.
For simplicity, we will note the approximation ratio of $\mathcal{A}_{qbST}$ as
$r_{qbST}$.

Our algorithms works in two phases.
The first phase computes a connected maximal matching in every connected component
of $G[V\bs S]$.
The second phase connects the remaining connected components by using the algorithm
$\mathcal{A}_{qbST}$ on quasi-bipartite instances of the Steiner Tree problem.\\
\noindent We formulate now an algorithm $\mathcal{A}_{SCVC}$ for the SCVC problem.
\begin{figure}[h]
  \fbox{\hspace*{.4cm}
    \begin{minipage}{\textwidth - .9cm}
      {\bf Algorithm $\mathcal{A}_{SCVC}$}\\[1.5ex]
      {\bf Input:} $G=(V,E)$ and $S\subseteq V$\\[1ex]
      \ding{192} 
       $S_1:=S$ and $S_2:=\emptyset$\\
      \hspace*{.5cm} {\bf while} $E(G[V\backslash S_1])\not=\emptyset$\\
      \hspace*{1cm}  compute a CMM $M_c$ in $G[V\backslash S_1]$\\
      \hspace*{1cm}  starting with a vertex $c\in V\backslash S_1$ which is connected to a $s\in S$ \\
      \hspace*{1cm}  $S_1:=V(M_c)\cup S_1$\\
      \hspace*{.5cm} {\bf endwhile}\\
      \ding{193}  contract every connected component $C$ of $G[S_1]$ into a vertex $s_C$\\
       \hspace*{.5cm}$S_2:=\{s_C\mid C \textrm{ is a connected component of } G[S_1]\}$\\
       \hspace*{.5cm}Let $G'$ be the graph after vertex contraction in $G$\\
      \ding{194} compute a Steiner Tree  $T$ for $S_2$ in $G'$ using $\mathcal{A}_{qbST}$\\
      {\bf Return}  $S_1\cup (V(T)\backslash S_2)$\\
      \medskip
    \end{minipage}
  }
  \caption{Algorithm $\mathcal{A}_{SCVC}$.}\label{figscvc}
\end{figure}

We are ready to prove the following
\begin{theorem}
Given a connected graph $G=(V,E)$ and $S\subset V$, the algorithm 
$\mathcal{A}_{SCVC}$ (see Figure~\ref{figscvc})
has an approximation ratio at most
$\max\left\{r_{qbST},\frac{2}{1+\frac{|S|}{|V|}}\right\}\;$.
\end{theorem}

\begin{proof}
Let $OPT$ be an optimal connected vertex cover of $G$ such that $S\subseteq OPT$.
Furthermore, let $S^f_1$ be the set $S_1$ at the beginning of phase \ding{193}
and $C_a$ be the set of connected components of 
$G[S^f_1]$. Since there exists for every connected component $C\in C_a$ a $s_C\in S\cap C$,
 we know that $|C\cap OPT|\geq 1$. Furthermore, it implies that
the cardinality of $C_a$ is at most $|S|$. Hence, we have to connect 
$|S_2|\leq |S|$ terminals in the graph $G'$. Since we computed a connected maximal matching in every 
connected component of $G[V\backslash S]$, the remaining vertices $V\backslash S^f_1$
form an independent set in $G$ and therefore, they must form an independent set in $G'$.
Thus, we can use an approximation algorithm for the Steiner Tree problem restricted to
quasi-bipartite instances to connect all the vertices in $S_2$.\\

In order to analyze the cost of the new vertices introduced in the phase \ding{194},
we will use an auxiliary graph construction:
For every $C\in C_a$, we contract the set $C\cap OPT$ into $s'_C$ in $G$. Let $G_{OPT}$ 
be the resulting graph and $S'_{OPT}=\{s'_C\mid C\in C_a\}$. 
Since we have $|C|\geq |C\cap OPT|\geq 1$ for every $C\in C_a$,
the cost of a minimum Steiner Tree $T_{OPT}$ for $S'_{OPT}$ in $G_{OPT}$ cannot be lower than the cost 
of a minimum Steiner Tree $T'$ for $S_2$ in $G'$. We want to lower bound $|OPT|$
by dividing the total cost into the cost per component in $C_a$ and the cost to connect the 
contracted components. Hence, we get $|OPT|\geq \sum_{C\in C_a} |C\cap OPT|+|V(T_{OPT})\bs S'_{OPT}|$.
Clearly, $\sum_{C\in C_a} |C\cap OPT|$ can be lower bounded by $|S|$ and the number of edges of the
computed connected maximal matchings.
Finally, we introduce $I=V(T')\backslash S_2$.

We are ready now to analyze the approximation ratio $R$ of the algorithm $\mathcal{A}_{SCVC}$:
\begin{eqnarray}
R &\leq & \frac{|S^f_1|+|V(T)\backslash S_2|}{|S^f_1\backslash S|/2+|S|+|V(T_{OPT})\backslash S'_{OPT}|}\nonumber\\
&\leq & \frac{|S^f_1|+\{cost(T)+1-|S_2|\}}{|S^f_1\backslash S|/2+|S|+\{cost(T_{OPT})+1-|S'_{OPT}|\}}\nonumber\\
&\leq & \frac{|S^f_1\backslash S|+|S|+cost(T)+1-|S_2|}{|S^f_1\backslash S|/2+|S|+cost(T')+1-|S_2|} \quad\quad\quad\quad\quad(*)\nonumber\\
&\leq & \frac{|S^f_1\backslash S|+|S|+r_{qbST}(|S_2|+|I|-1)+1-|S_2|}{|S^f_1\backslash S|/2+|S|+(|S_2|+|I|-1)+1-|S_2|}\nonumber\\
&\leq & \frac{|S^f_1\backslash S|+|S|+(r_{qbST}-1)(|S_2|-1)+r_{qbST}|I|}{|S^f_1\backslash S|/2+|S|+|I|}\nonumber\\
&\leq & \frac{|S^f_1\backslash S|+|S|+(r_{qbST}-1)(|S|-1)+r_{qbST}|I|}{|S^f_1\backslash S|/2+|S|+|I|}\nonumber\\ 
&\leq & \frac{|S^f_1\backslash S|+r_{qbST}(|S|+|I|)}{|S^f_1\backslash S|/2+|S|+|I|} =: F \nonumber 
\end{eqnarray}
In $(*)$, we used the facts that $cost(T_{OPT})\geq cost(T') $ and $|S_2|=|S'_{OPT}|=|C_a|$.\\

We have to analyze the worst case behavior of the fraction $F$. Clearly, $F$ cannot be smaller than $r_{qbST}$. On the other hand,     
since $|S^f_1\backslash S|$ is the term with the largest coefficient in $F$, we can maximize 
$F$ by maximizing $|S^f_1\backslash S|$. Therefore, we 
set $|OPT|=|S|+|S^f_1\backslash S|/2$ and $|V|=|S|+|S^f_1\backslash S|$. In due consideration of the constraints, we derive
\begin{eqnarray*}
F & \leq & \max\left\{r_{qbST},\frac{|S|+|S^f_1\backslash S|}{|S|+|S^f_1\backslash S|/2}\right\}
    \, = \, \max\left\{r_{qbST},\frac{2}{\frac{|S|+|S^f_1\backslash S|/2}{|S|/2+|S^f_1\backslash S|/2}}\right\} \\
  & = & \max\left\{r_{qbST},\frac{2}{1+\frac{|S|}{|V|}}\right\}\,.
\end{eqnarray*}
\end{proof}

\subsection{Application to the Connected Vertex Cover Problem}

By combining the modified II-algorithm with $\mathcal{A}_{SCVC}$, we obtain the following approximation result.
\begin{theorem}
There is a randomized approximation algorithm with an approximation factor arbitrarily close to $2/(1+\bar{d}/(2\Delta))$ for the Vertex Cover problem in graphs with average degree $\bar{d}$, 
maximum degree $\Delta\leq n/2$ and $\Delta=\Omega(n/\psi(n))$ 
with
running time $n^{O(1)} 2^{O(\psi(n)\log\log\log(n))}$ which 
is
\begin{itemize}
\item polynomial if $\psi (n) = O\left(\frac{\log n}{\log \log \log n}\right)$,
\item quasi-polynomial if $\psi (n) = \polylog n$.
\end{itemize}
\end{theorem}
\begin{proof}
It is not hard to see that the original II-algorithm can generate 
with high probability a part $W$ of an optimal connected vertex cover,
since a connected vertex cover is a vertex cover with additional properties
and the II-algorithm only makes use of the necessary requirements of a vertex cover.
Since the modified II-algorithm takes advantage of the same requirements,
we can invoke the modified II-algorithm to generate with high probability a part $W$ of an optimal 
connected vertex cover with cardinality $|W|\geq\frac{\bar{d}}{2\Delta}n$ (cf. Theorem~\ref{thm:mainvc}).

Then, we apply the algorithm $\mathcal{A}_{SCVC}$ to $G$ and $S=W$ to obtain a connected vertex cover of $G$. 
For the considered range of values for $\Delta$, this combined algorithm has an approximation ratio at most 
$\max\left\{r_{qbST}, \frac{2}{1+\left. \frac{\bar{d}}{2\Delta}\right. }\right\}$. But the latter ratio is always at least $4/3$, and thus is 
always greater than $r_{qbST} \leq 1.217$ \, (cf.~\cite{GHNP02}).
\end{proof}


\section{The Subdense Set Cover Problem}
\label{sec:sc}

Recall that a set cover instance consists of a ground set $X$ of size $n$ and a collection $P\subset 2^X$ of subsets of $X$, with $|P|=m$. We wish to find a minimum subset $M\subseteq P$ such that $\bigcup_{S\in M} S = X$. We consider subdense instances of the Set Cover problem, generalizing the dense version of the Set Cover problem proposed by Karpinski and Zelikovsky~\cite{KZ97}. In the $\epsilon$-dense Set Cover problem, every element in $X$ belongs to at least a fraction $\epsilon$ of the subsets in $P$. We call an instance of the Set Cover problem $\psi (n)$-dense whenever every element is contained in at least $m/\psi (n)$ subsets. 

It is well-known that the greedy algorithm provides a $(1 + \ln n)$-approximate solution to the general Set Cover problem, and that a logarithmic approximation is essentially the best achievable, under classical complexity-theoretic assumptions~\cite{F98}. The greedy algorithm picks at every iteration the subset that covers the largest number of remaining (uncovered) elements. We now consider the behavior of this algorithm for the $\psi (n)$-dense instances of set cover.

\begin{lemma}
\label{lem:greedysubdensesc}
The greedy algorithm for $\psi (n)$-dense instances of the Set Cover problem returns a solution of size at most $\psi (n) \ln n$.
\end{lemma}
\begin{proof}
At the first iteration, every element in $X$ is contained in at least $m/\psi (n)$ subsets. Hence, by the pigeonhole principle, there exists a subset covering at least $n/\psi (n)$ elements of $X$. After the $i$-th iteration of the algorithm, there are exactly $m-i$ subsets remaining, but every uncovered element is still contained in at least $m/\psi (n) \geq (m-i)/\psi (n)$ remaining subsets, as otherwise it would have been covered in a previous step. Hence, there remains a subset covering at least a fraction $1/\psi (n)$ of the uncovered elements. Thus, we reduce the number of uncovered elements by a factor $1-\frac 1{\psi (n)}$ at every iteration. 

The number of iterations required is therefore at most
$$
\ln n \left[\ln\left(\frac{1}{1-\frac{1}{\psi (n)}}\right)\right]^{-1} .
$$
We assume that $\psi(n)\geq 2$ and proceed by deriving an upper bound on the second term: 
\begin{eqnarray*}
\left[\ln\left(\frac{1}{1-\frac{1}{\psi (n)}}\right)\right]^{-1}
  & = & \left[\ln\left(\frac{\psi (n)}{\psi (n) - 1}\right)\right]^{-1} \\
  & = & \frac{1}{\ln (\psi(n)) - \ln (\psi (n)-1)} \\
  & = & \left( \int_{\psi (n)-1}^{\psi (n)} \frac1x dx \right)^{-1} \\
  & < & \left( (\psi (n) - (\psi (n)-1))\frac1{\psi (n)}\right)^{-1} \\
  & = & \psi (n) ,
\end{eqnarray*}
which proves the lemma.
\end{proof}

From this bound, we know that we can solve the problem exactly by performing an exhaustive search in time $O\left( m^{\ln n \cdot \psi(n)}\right)$. For sufficiently small functions $\psi (n)$, this is subexponential. Hence, we obtain that unless $NP\subseteq DTIME [n^{\log n \cdot \psi(n)}]$, the $\psi(n)$-dense Set Cover problem is not NP-complete. This generalizes the analogous observation for the dense Set Cover problem previously made by Karpinski and Zelikovsky~\cite{KZ97}.


\section{The Subdense Steiner Tree Problem}
\label{sec:st}

In this section, we study the approximability of nondense instances of the Steiner Tree problem. 
In particular, we show that $O(\psi(n))$-dense instances of the Steiner Tree problem 
can be approximated to within $1+\delta$ in time $n^{O(1)}2^{O(\frac{\psi(n)}{\delta})}$.
On the other hand, we show that for every $\delta,\epsilon>0$, the Steiner Tree problem restricted to 
$|V\backslash S|^{1-\delta}$-everywhere dense
graphs is NP-hard to approximate with $\frac{263}{262}-\epsilon$.\\

Karpinski and Zelikovsky \cite{KZ97} introduced the $\epsilon$-dense Steiner Tree problem where an instance
is called $\epsilon$-dense if each terminal $t\in S$ has at least $\epsilon|V\backslash S|$ 
neighbors in $V\backslash S$. They showed that for every $\epsilon>0$ the $\epsilon$-dense
Steiner Tree problem admits a PTAS. On a high level, the scheme finds in every step
a non-terminal which is adjacent to a large fraction of the actual terminal set, thus reducing
the cardinality of the terminal set to a constant after iterating this process. The remaining
instance can be solved by an exact algorithm in polynomial time. Hauptmann~\cite{H07} showed that
the same scheme actually yields even an EPTAS for the $\epsilon$-dense Steiner Tree problem.\\

An instance of the Steiner Tree problem on graphs consisting of $G=(V,E)$  
and terminal set $S$ is called $\psi(n)$-dense if each terminal has at least $\frac{|V\backslash S|}{\psi(n)}$
neighbors in $|V\backslash S|$. 

We show that the Karpinski-Zelikovsky-techniques can be extended to subdense cases and prove the following
\begin{theorem}\label{thm:stpmain}
$\psi(n)$-dense instances of the Steiner Tree problem 
can be approximated to within $1+\delta$ in time $n^{O(1)}2^{O(\frac{\psi(n)}{\delta})}$.
\end{theorem} 
By setting $\psi(n)$ to $O(\log n)$, we obtain the existence of PTAS for subdense instances of 
the Steiner Tree problem where every terminal has at least $\Omega(\frac{|V\backslash S|}{\log n})$
neighbors in $V\backslash S$.
\begin{corollary}
There exists a PTAS for $O(\log n)$-dense instances of the Steiner Tree problem. 
\end{corollary}
By setting $\psi(n)$ to $\polylog(n)$, it yields quasi-polynomial time approximation schemes.
\begin{corollary}
There exist quasi-polynomial time approximation schemes (QPTAS) for $\polylog(n)$-dense instances of the 
Steiner Tree problem.
\end{corollary}

\begin{figure}
  \fbox{\hspace*{.4cm}
    \begin{minipage}{\textwidth - .9cm}
      {\bf Algorithm $\mathcal{M}_{DSTP}$}\\[1.5ex]
      {\bf Input:} $\psi(n)$-dense instance $G=(V,E)$, $S\subseteq V$ and $\delta>0$\\[1ex]
      \ding{192} 
       $S_1:=S$\\
      \hspace*{.5cm} {\bf while} $|S_1|>s(\delta,\psi):=\max\{\frac{2}{\delta},2\}\cdot \psi(n)$\\
      \hspace*{1cm}  $v:=\text{argmax}\left\{|N(u)\cap S_1|,u\in V\backslash (S\cup N)\right\}$\\
      \hspace*{1cm}  $S(v):=\text{ be the star consisting of } v \text{ and } N(v)\cap S_1$\\
      \hspace*{1cm}  $N:=N\cup\{v\}$ and contract $S(v)$ into $s_v$\\
      \hspace*{1cm}  $S_1:=(S_1\backslash N(v))\cup\{s_v\}$\\
      \hspace*{.5cm} {\bf endwhile}\\
      \ding{193}
       Apply Dreyfus-Wagner algorithm to remaining instance to obtain $T_{D\&W}$ \\[1ex]
      {\bf Return} $T_{\delta}:=T_{D\&W}\cup\bigcup_{v\in N}S(v)$\\
      \medskip
    \end{minipage}
  }
  \caption{Modified DSTP algorithm.}
\end{figure}

\paragraph{Analysis of the Algorithm $\mathcal{M}_{DSTP}$\\}
First of all, we will prove a lemma of significant importance to our analysis
and proof of our main result.
It shows that we can extract a large portion of the actual terminal set $S_1$. 

\begin{lemma}\label{lemma:extract}
In every iteration of the phase \ding{192}, the cardinality of the extracted neighborhood contained in $S_1$
can be lower bounded by $|N(v)\cap S_1|\geq \frac{|S_1|}{2\psi(n)}$.
\end{lemma}
\begin{proof}
We define $S^i_1$ as the set $S_1$ after the $i$-th iteration of phase \ding{192},
$S(v_i)$ be the star picked in the $i$-th iteration and $N^i$ be the set of non-terminals after
the $i$-th iteration.
In addition, we introduce a data structure called \emph{contraction tree},
 which will be one of the main tools in our proof.
 
\begin{definition}[Contraction Tree]
 Given a terminal $s_{v_j}\in S^i_1\backslash S$ with $i>0$, 
we define recursively the contraction tree $T(s_{v_j})$. The root of $T(s_{v_j})$ 
consists of the node $s_{v_j}$ and the child nodes of $s_{v_j}$ are defined as $c\in S^{j}_1$ where $j<i$ and $S^j_1$ is the set that was contracted into $s_{v_j}$. 
All nodes of $T(s_{v_j})$ contained in $S$ are leaf nodes. The subtree of $T(s_{v_j})$ at 
$c\in  S^{j}_1\backslash S$ is given by $T(c)$.
\end{definition}

  First of all, we prove a simple fact concerning the number of non-terminal neighbors of nodes in $S_1^i\bs S$. 
\begin{fact}
Given a node $s_{v_j}\in S^i_1\backslash S$, the number of non-terminal neighbors of
$s_{v_j}$ can be lower bounded by $\frac{|V\backslash S|}{\psi(n)}-height(T(s_{v_j}))$.
\end{fact}
\begin{proof}
Let $s_{v_k}\in V(T(s_{v_j}))\bs S$ be a inner node of $T(s_{v_j})$ and
 $S_1^k$ be the set which was contracted into $s_{v_k}$.
Let $s'\in S_1^k$ be the node with $d$ non-terminal neighbors where 
$d=\min_{s\in S_1^k}\{\#\textrm{non-terminal neighbors of }s\}$.
 Since every contraction involves only one non-terminal,
  the number of non-terminal neighbors of the parent node $s_{v_k}$ of $s'$ is 
at least $d-1$. Therefore, we can conclude that the root $s_{v_j}$ must have at least 
  $\frac{|V\backslash S|}{\psi(n)}-height(T(s_{v_j}))$ non-terminal neighbors.
\end{proof}
Clearly, the remaining non-contracted terminals still have at least $\frac{|V\backslash S|}{\psi(n)}$ neighbors.
 Therefore, the number of edges between $S_1^i$ and $N^i$ denoted as 
$|E(S_1^i,N^i)|$ can be lower bounded by
\begin{eqnarray*}
|E(S_1^i,N^i)|&\geq& (|S^i_1\cap S|)\frac{|V\backslash S|}{\psi(n)}+
\sum_{t\in S^i_1\backslash S} |N(t)\cap (N^i)| \\
&\geq& (|S^i_1|-|S^i_1\backslash S|)\frac{|V\backslash S|}{\psi(n)}+
\sum_{t\in S^i_1\backslash S}\frac{|V\backslash S|}{\psi(n)}-height(T(t))\\
&\geq & |S^i_1|\frac{|V\backslash S|}{\psi(n)}-\sum_{t\in S^i_1\backslash S}height(T(t))\\
& \geq & |S^i_1|\frac{|V\backslash S|}{\psi(n)}-\sum_{t\in S^i_1\backslash S}|\{s_l\mid 
s_l \textrm{ is a inner node of } T(t)\}|
\end{eqnarray*}  
Since the inner nodes of the contraction tree corresponds to contractions and therefore,
corresponds to iterations, we get $$\sum_{t\in S^i_1\backslash S}|\{s_l\mid 
s_l \textrm{ is a inner node of } T(t)\}|\leq i \textrm{ and }
|E(S_1^i,N^i)|\geq  |S^i_1|\frac{|V\backslash S|}{\psi(n)}-i.$$\\
After iteration $i$ of the phase \ding{192}, we still have $|N^i|=|V\backslash S|-i$ non-terminals. We can assume that
$|N^i|>0$ since otherwise we are done. Recall that $|S_1|\geq 2\psi(n)$. By pigeonhole principle,
we know there must be a $v\in N^i$ with the following number of neighbors in $S^i_1$:
$$\frac{|E(S_1^i,N^i)|}{|N^i|}\geq \frac{|S^i_1|\frac{|V\backslash S|}{\psi(n)}-i}{|V\backslash S|-i}\geq 
\frac{|S^i_1|\frac{|V\backslash S|}{\psi(n)}-i\frac{|S^i_1|}{\psi(n)}}{|V\backslash S|-i} \geq \frac{|S^i_1|}{\psi(n)}$$
Thus, we can pick in every iteration at least $\frac{|S_1|}{\psi(n)}$ terminals.

But since we contract a star $S(v)$ into $s_v$ and add $s_v$ to $S_1$, we decreased $S_1$ only by
$\frac{S_1}{\psi(n)}-1$. Alternatively, this can be seen as a greedy pick with a slightly worse density 
$c\cdot \psi(n)$. In other words, we want to show that there is a $c>1$ such that
$\frac{S_1}{\psi(n)}-1\geq \frac{S_1}{c\cdot\psi(n)}$. After short calculation, we get 
$|S_1|(c-1)\geq c\cdot \psi(n)$. Since the cardinality of $S_1$ is at least $2\cdot \psi(n)$,
we obtain $\psi(n) \geq \frac{c}{c-1}\psi(n)$, which can be ensured by $c:=2$.
By combining these two arguments, we conclude
 that we can reduce the set $S_1$ by $\frac{|S_1|}{2\psi(n)}$ in every iteration.
\end{proof}
 
The next lemma deals with the upper bound on the cardinality of the set $N$
at the beginning of phase \ding{193}.

\begin{lemma}
The cardinality of the set $N$ at the end of phase \ding{192} 
can be upper bounded by 
$2\ln(\frac{|S|}{s(\delta,\psi(n))})\cdot\psi(n)$.
\end{lemma} 

\begin{proof}
According to Lemma~\ref{lemma:extract} , we can pick in every iteration at least $\frac{|S_1|}{2\psi(n)}$
terminals.
If we define $S^i_1$ as the set $S_1$ after the $i$-th iteration of the phase \ding{192},
we see that $|S^i_1|\leq \left(1-\frac{1}{2\psi(n)}\right)^i|S|$. In order to decrease $|S^i_1|$
to $s(\delta, \psi(n))$, we have to iterate the first phase at least 
$\ln\left(\frac{|S|}{s(\delta, \psi(n))}\right)\left[\ln\left(\frac{1}{1-\frac{1}{2\cdot\psi(n)}}\right)\right]^{-1}$ times. 
By applying lemma~\ref{lem:greedysubdensesc}, we obtain the following inequality:
$$|N|=\ln\left(\frac{|S|}{s(\delta, \psi(n))}\right)\left[\ln\left(\frac{1}{1-\frac{1}{2\cdot\psi(n)}}\right)\right]^{-1}< 
\ln\left(\frac{|S|}{s(\delta, \psi(n))}\right)2\cdot \psi(n)$$
\end{proof}

We are ready to prove our main Theorem~\ref{thm:stpmain}.
\begin{proof}
Let $T_{OPT}$ be an optimal Steiner Tree for $S$ in $G$ and $T_2$ be the optimal
steiner tree for $S_1$ computed in the phase \ding{193}. For every $v\in N$, we add a set $E(T_v)$ of 
edges  to $E(G)$ that form a spanning tree $T_v$ for $S(v)\backslash\{v\}$. Let 
$G':=(V(G),E(G)\cup \bigcup_{v\in N} E(T_v)$ and $T'_{OPT}$ be an optimal steiner 
tree for $S$ in $G'$. Clearly, we have $c(T_{OPT})\geq c(T'_{OPT})$ since adding edges to a graph
cannot increase the cost of an optimal steiner tree. By definition of $T_{D\&W}$, we obtain 
$c(T'_{OPT})=c(T_{D\&W})+\sum_{v\in N}c(T_v)$. Hence, we can relate the cost of $T_{\delta}$
with the cost of $T'_{OPT}$ in order to analyze the approximation ratio $R$ of the algorithm:
\begin{eqnarray*}
R & \leq &
\frac{c(T_{\delta})}{c(T'_{OPT})}
 \leq 
\frac{\sum_{v\in N}(|S(v)|-1)+c(T_{D\&W})  }{\sum_{v\in N}(|S(v)\backslash \{v\}|-1)+c(T_{D\&W})} \\
& \leq & 1+\frac{|N|}{\sum_{v\in N}(|S(v)\backslash \{v\}|-1)+c(T_{D\&W})} \\
& \leq &
1+\frac{|N|}{|S|}<1+  \frac{\ln\left(\frac{|S|}{s(\delta,\psi(n))}\right)2\cdot\psi(n)}{|S|}
\leq 1+  \frac{\ln\left(\frac{|S|}{s(\delta,\psi(n))}\right)\delta}{\frac{|S|}{s(\delta,\psi(n))}}
\leq 1+\delta \\
\end{eqnarray*}

\noindent The running time of the Dreyfus-Wagner algorithm~\cite{DW71}~is~$O(3^{|S|}n\!+\!2^{|S|}n^2\!+\!n^3)$.
Hence, the running time of the algorithm $\mathcal{M}_{DSTP}$ is dominated by $n^{O(1)}2^{O(s(\delta,\psi(n)))}=n^{O(1)}2^{O(\frac{\psi(n)}{\delta})}$.
\end{proof}

\subsection{Inapproximability Results}
We prove now the following approximation hardness result for the Steiner Tree problem
restricted to graphs where the degree of every node is  at least $|V\backslash S|^{1-\delta}$.

\begin{theorem}
\label{thm:stplb}
For every $\delta,\epsilon>0$, the Steiner Tree problem restricted to 
$|V\backslash S|^{1-\delta}$-everywhere dense
graphs is APX-hard and NP-hard to approximate with $\frac{263}{262}-\epsilon$.
\end{theorem}

\begin{proof}
Berman and Karpinski \cite{BK03} obtained explicit inapproximability results for the Vertex Cover problem
in graphs with bounded maximum degree. We will combine this explicit lower bound with an approximation 
preserving reduction from the Vertex Cover problem restricted to bounded degree graphs to the Steiner
Tree problem. In particular, we will create a special instance of the Steiner Tree problem where
the degree of every terminal is exactly $2$. It will be the starting point of our densification process
to obtain explicit hardness results on graphs with high degree.\\

First, we describe the approximation preserving reduction due to Bern and Plassman \cite{BP89}:
Given a graph $G=(V,E)$ with maximum degree $B$ as an instance of the Vertex Cover problem,
we construct the graph $G_{ST}=(V_{ST},E_{ST})$ and the terminal set $S_{ST}\subseteq V_{ST}$.

For every edge $e\in E$, we introduce a vertex $v_e$ and join it with the two vertices
$w,x\in e$ in $G_{ST}$. Hence, we get $V_{ST}=V\cup \{v_e\mid e\in E\}$. Finally,
we connect all $x,y\in V$ in $G_{ST}$ and define $S_{ST}:=\{v_e\mid e\in E\}$. 
Clearly, the degree of every terminal is $2$ and the length of an optimal Steiner Tree is 
$|E|+|VC_{OPT}|-1\leq \frac{B}{2}|V|+|VC_{OPT}|-1$ where $VC_{OPT}$ denotes a minimum
vertex cover in $G$. We are ready to deduce the explicit lower bound for the special instance
of the Steiner Tree problem. 

The following hardness result for bounded degree version of the 
Vertex Cover problem is due to Berman and Karpinski:
\begin{theorem}[Berman and Karpinski~\cite{BK03}]
The Vertex Cover problem is NP-hard to approximate with $\frac{55}{54}-\epsilon$ in graphs $G$ 
with maximum degree $\Delta_G=4$.
\end{theorem}
More precisely, they proved that for every $\epsilon\in (0,1/2)$, it is NP-hard to decide whether an instance
of the Vertex Cover problem restricted to graphs with maximum degree $\Delta_G=4$ with $104n$ nodes 
has a vertex cover of size below $(54+\epsilon)n$ or above $(55-\epsilon)n$.\\   

Combining these results with the Bern-Plassman reduction, we obtain the following:
The Berman-Karpinski graph for the Vertex Cover problem restricted to graphs with maximum 
degree $\Delta_G=4$ has at most $4\cdot 104n/2=208n$ edges. Therefore, it is
is NP-hard to decide whether $smt\leq 208n+(54+\epsilon)n-1=(262+\epsilon)n-1$ or 
$smt\geq 208n + (55-\epsilon)n-1=(263-\epsilon)n-1$. Hence, the described special instance of the Steiner Tree
problem is NP-hard to approximate with approximation ratio $\frac{263}{262}-\epsilon$.

Now, we need to increase the degree of every terminal. In order to obtain an instance with high vertex degree,  
we associate with every non-terminal $n\in V_{ST}\backslash S_{ST}$
a set $U_n$ of $k$ new vertices and introduce $U:=\bigcup_{n\in V_{ST}\backslash S_{ST}} U_n$. We construct the 
graph $D=(V(D),E(D))$ with $V(D)=U\cup S_{ST}$ and $E(D):=\{\{s,u\}\mid u\in U_x, s\in S_{ST},
\{s,x\}\in E_{ST}\}\cup {U \choose 2} $. The terminal set of this instance remains the same $S_{D}:=S_{ST}$.
To determine the vertex degree of a terminal, we define $m:=|V_{ST}\backslash S_D|$.
Furthermore, we choose $k:=m^{\frac{1-\delta}{\delta}}$. 
Notice that the degree
of every $s\in S$ is  at least 
$2k=2m^{\frac{1-\delta}{\delta}}=2(m^{\frac{1-\delta}{\delta}}\cdot m)^{1-\delta}=2|U|^{1-\delta}=2|V(D)\backslash S_D|^{1-\delta}$.
Next, we analyze how the optimal solutions for $(G_{ST},S_{ST})$ and $(D,S_{D})$ can be related to
each other. Let $OPT_{ST}$ [$OPT_D$] be an optimal solution for $(G_{ST},S_{ST})$ [$(D,S_{D})$].
Since it does not improve the cost of a feasible solution for $(D,S_{D})$ to have more than one vertex of every 
$U_x$ with $x\in V_{ST}\backslash S_{ST}$, we have to consider only feasible solutions 
$T=(V_{T},E_T)$ with $|U_x\cap V_T |\leq 1$ for every $x\in V_{ST}\backslash S_{ST}$. This implies
a canonical and cost preserving reduction. In particular, we have $|OPT_{ST}|=|OPT_D|$. 
All in all, we obtain for every $\epsilon>0$ the same inapproximability factor $\frac{263}{262}-\epsilon$  for instances where every vertex has degree at least $|V(D)\backslash S_D|^{1-\delta}$.
\end{proof}

\section{Summary of Results}
Tables~\ref{tab:res:upp}--\ref{tab:res:troffs} give now a summary of our main results.

\begin{table}[h]
\centering
\begin{tabular}{||c|c|c||}
  \hline \hline
  & & \\
  \vcent{\textbf{Problem}}  &
  \vcent{\textbf{Upper Bound}} & 
  \vcent{\textbf{Approx. Hardness}} \\
  \hline
  & & \\
  \vcent{VC} &
  \vcent{$\frac{2}{1+\bar{d}/(2\Delta)}$} &
  \vcent{$\frac{2}{1+\bar{d}/(2\Delta)}-\delta$  (UGC) \cite{II05}}  \\
  \hline                
  & & \\
  \vcent{CVC} &
  \vcent{$\frac{2}{1+\bar{d}/(2\Delta)}$} &
  \vcent{$\frac{2}{1+\bar{d}/(2\Delta)}-\delta$  (UGC) \cite{II05}} \\
  \hline
  & & \\
  \vcent{STP} &
  \vcent{PTAS} &
  \vcent{$-$} \\
  \hline
  \hline
\end{tabular}
\caption{Table of subdense results\label{tab:res:upp}}
\end{table}

\begin{table}[h]
\centering
\begin{tabular}{||c|c|c|c||}
  \hline \hline
  & & &\\
  \vcent{\textbf{Problem}}  &
  \vcent{\textbf{Upper Bound}} &
  \vcent{\textbf{Approx. Hardness}} &
  \vcent{\textbf{Run Time}} \\
  \hline
  & & &\\[1ex]
  \vcent{CVC} &
  \vcent{$\frac{2}{1+\bar{d}/(2\Delta)}$} &
  \vcent{$\frac{2}{1+\bar{d}/(2\Delta)}-\delta$  (UGC) \cite{II05}}  &
  \vcent{quasipoly time} \\
  \hline
  & & & \\
  \vcent{STP} &
  \vcent{PTAS} &
  \vcent{$\frac{2}{1+\bar{d}/(2\Delta)}-\delta$  (UGC) \cite{II05}}  &
  \vcent{quasipoly time} \\
  \hline
  \hline
\end{tabular}
\caption{Table of mildly sparse approximability results }
\end{table}

\begin{table}[h]
\centering
\begin{tabular}{||c|c|c||}
  \hline
  \hline
  & &\\
  \vcent{\textbf{Problem}} &
  \vcent{\textbf{Approx. Ratio}} &
  \vcent{\textbf{Run Time}} \\
  \hline
  & & \\
  \vcent{VC} &
  \vcent{$\frac{2}{1+\bar{d}/(2\Delta)}$} &
  \vcentmath{$n^{O(1)}2^{O\left(\psi(n)\ln\ln\ln n\right)}$} \\
  \hline                
  & & \\
  \vcent{CVC} &
  \vcent{$\frac{2}{1+\bar{d}/(2\Delta)}$} &
  \vcentmath{$n^{O(1)}2^{O\left(\psi(n)\ln\ln\ln n\right)}$} \\
  \hline
  & & \\
  \vcent{STP} &
  \vcent{$1+\delta$}  & 
  \vcentmath{$n^{O(1)}2^{O(\frac{\psi(s)}{\delta})}$} \\
  \hline
  & & \\
  \vcent{SC} &
  \vcent{exact} &
  \vcentmath{$O(m^{\psi(n)\ln n})$}   \\  
  \hline                 
  \hline                 
\end{tabular}
\caption{Table of known tradeoffs\label{tab:res:troffs}}
\end{table}


\section{Conclusion}

We have established several approximability results for the subdense
instances of Covering and Steiner Tree problems, proving in some cases 
optimality of the approximation ratios under standard complexity-theoretic
assumptions. A very interesting open problem still remains the status
of the dense (and of course also of the subdense) Steiner Forest problem. Our methods 
do not apply to this problem directly, and some new techniques seem to be
necessary.


\end{document}